\documentclass[11pt]{article}%
\usepackage{amsfonts}
\usepackage{amsmath}
\usepackage{amssymb}
\usepackage{graphicx}%
\providecommand{\U}[1]{\protect\rule{.1in}{.1in}}
\newtheorem{theorem}{Theorem}
\newtheorem{acknowledgement}[theorem]{Acknowledgement}

\newtheorem{corollary}[theorem]{Corollary}

\newtheorem{lemma}[theorem]{Lemma}
\newtheorem{notation}[theorem]{Notation}

\newtheorem{proposition}[theorem]{Proposition}

\newenvironment{proof}[1][Proof]{\noindent\textbf{#1.} }{\ \rule{0.5em}{0.5em}}
\begin{document}

\title{Born--Jordan Quantization and the Uncertainty Principle}
\author{Maurice A. de Gosson\\University of Vienna\\Faculty of Mathematics (NuHAG)\\Nordbergstr. 15, 1090 Vienna}
\maketitle

\begin{abstract}
The Weyl correspondence and the related Wigner formalism lie at the core of
traditional quantum mechanics. We discuss here an alternative quantization
scheme, whose idea goes back to Born and Jordan, and which has recently been
revived in another context, namely time-frequency analysis. We show that in
particular the uncertainty principle does not enjoy full symplectic covariance
properties in the Born and Jordan scheme, as opposed to what happens in the
Weyl quantization.

\end{abstract}

\section{Introduction}

The problem of \textquotedblleft quantization\textquotedblright\ of an
\textquotedblleft observable\textquotedblright\ harks back to the early days
of quantum theory; mathematically speaking, and to use a modern language, it
is the problem of assigning to a symbol a pseudo-differential operator in a
way which is consistent with certain requirements (symmetries under a group of
transformations, positivity, etc.). Two of the most popular quantization
schemes are the Kohn--Nirenberg and Weyl correspondences. The first is widely
used in the theory of partial differential equations and in time-frequency
analysis (mainly for numerical reasons), the second is the traditional
quantization used in quantum mechanics. Both are actually particular cases of
Shubin's pseudo-differential calculus, where one can associate to a given
symbol $a$ an infinite family $(A_{\tau})_{\tau}$ of pseudo-differential
operators parametrized by a real number $\tau$, the cases $\tau=1$ and
$\tau=\frac{1}{2}$ corresponding to, respectively, Kohn--Nirenberg and Weyl
operators. It turns out that each of Shubin's $\tau$-operators can be
alternatively defined in terms of a generalization $\operatorname*{Wig}_{\tau
}$ of the usual Wigner distribution by the formula%
\[
\langle A_{\tau}\psi,\overline{\phi}\rangle=\langle a,\operatorname*{Wig}%
\nolimits_{\tau}(\psi,\phi)\rangle
\]
and this observation has recently been used by researchers in time-frequency
analysis to obtain more realistic phase-space distributions (more about this
in the discussion at the end of the paper). They actually went one step
further by introducing a new distribution by averaging $\operatorname*{Wig}%
\nolimits_{\tau}$ for the values of $\tau$ in the interval $[0,1]$. This
leads, via the analogue of the formula above to a third class of
pseudo-differential operators, corresponding to the averaging of Shubin's
operators $A_{\tau}$. It was noted by the present author that this averaged
pseudo-differential calculus is actually an extension of one of the first
quantization schemes discovered by Born, Jordan, and Heisenberg around 1927,
prior to that of Weyl's.

The aim of the present paper is to give a detailed comparative study of the
Weyl and Born--Jordan correspondences (or \textquotedblleft quantization
schemes\textquotedblright) with an emphasis on the symplectic covariance
properties of the associated uncertainty principles.

\begin{notation}
We write $x=(x_{1},...,x_{n})$ and $p=(p_{1},...,p_{n})$ and $z=(x,p)$. In
matrix calculations $x,p,z$ are viewed as column vectors. The phase space
$\mathbb{R}^{2n}\equiv\mathbb{R}^{n}\times\mathbb{R}^{n}$ is equipped with the
standard symplectic form $\sigma(z,z^{\prime})=px^{\prime}-p^{\prime}x$;
equivalently $\sigma(z,z^{\prime})=Jz\cdot z^{\prime}$ where $J=%
\begin{pmatrix}
0_{n\times n} & I_{n\times n}\\
-I_{n\times n} & 0_{n\times n}%
\end{pmatrix}
$ is the standard symplectic matrix. We denote by $\mathcal{S}(\mathbb{R}%
^{2n})$ the Schwartz space of rapidly decreasing smooth functions and by
$\mathcal{S}^{\prime}(\mathbb{R}^{2n})$ its dual (the tempered distributions).
\end{notation}

\section{Discussion of Quantization}

After Werner Heisenberg's seminal 1925 paper \cite{heisenberg} which gave
rigorous bases to the newly born \textquotedblleft quantum
mechanics\textquotedblright, Born and Jordan \cite{bj} wrote the first
comprehensive exposition on matrix mechanics, followed by an article with
Heisenberg himself \cite{bjh}. These articles were an attempt to solve an
ordering problem: assume that some quantization process associated to the
canonical variables $x$ (position) and $p$ (momentum) two operators
$\widehat{X}$ and $\widehat{P}$ satisfying the canonical commutation rule
$\widehat{X}\widehat{P}-\widehat{P}\widehat{X}=i\hbar$. What should then the
operator associated to the monomial $x^{m}p^{n}$ be? Born and Jordan's answer
was%
\begin{equation}
x^{m}p^{n}\overset{\text{\textrm{BJ}}}{\longleftrightarrow}\frac{1}{n+1}%
\sum_{k=0}^{n}\widehat{P}^{n-k}\widehat{X}^{m}\widehat{P}^{k} \label{bj1}%
\end{equation}
which immediately leads to the \textquotedblleft symmetrized\textquotedblright%
\ operator $\frac{1}{2}(\widehat{X}\widehat{P}+\widehat{P}\widehat{X})$ when
the product is $xp$. In fact Weyl and Born--Jordan quantization lead to the
same operators for all powers $x^{m}$ or $p^{n}$, or for the product $xp$ (for
a detailed analysis of Born and Jordan's derivation see Fedak and Prentis
\cite{fepr09}, also Castellani \cite{ca78} and Crehan \cite{cr89}).
Approximately at the same time Hermann Weyl had started to develop his ideas
of how to quantize the observables of a physical system, and communicated them
to Max Born and Pascual Jordan (see Scholz \cite{scholz}). His basic ideas of
a group theoretical approach were published two years later
\cite{Weyl,WeylRob}. One very interesting novelty in Weyl's approach was that
he proposed to associate to an observable of a physical system what we would
call today a Fourier integral operator. In fact, writing the observable as an
inverse Fourier transform%
\begin{equation}
a(x,p)=\int_{\mathbb{R}^{2n}}e^{i(ps+xt)}\mathcal{F}a(s,t)dsdt \label{w1}%
\end{equation}
he defined its operator analogue by%
\begin{equation}
A=\int_{\mathbb{R}^{2n}}e^{i(\widehat{P}s+\widehat{X}t)}\mathcal{F}a(s,t)dsdt
\label{w2}%
\end{equation}
which is essentially the modern definition that will be given below (formula
(\ref{ahat}). We will denote the Weyl correspondence by $a\overset
{\text{\textrm{Weyl}}}{\longleftrightarrow}A_{\mathrm{W}}$ or $A_{\mathrm{W}%
}=\operatorname*{Op}(a)$. Weyl was led to this choice because of the immediate
ordering problems that occurred when one considered other observables than
monomials $a(x,p)=x^{k}$ or $a(x,p)=p^{\ell}$. For instance, using
Schr\"{o}dinger's rule what should the operator associated with $a(x,p)=xp$
be? Weyl's rule immediately yields the symmetrized\ quantization rule
\[
a(\widehat{X},\widehat{P})=\frac{1}{2}(\widehat{X}\widehat{P}+\widehat
{P}\widehat{X})
\]
and one finds that more generally (McCoy \cite{mccoy}, 1932) \
\begin{equation}
x^{m}p^{n}\overset{\text{\textrm{Weyl}}}{\longleftrightarrow}\frac{1}{2^{n}%
}\sum_{k=0}^{n}%
\begin{pmatrix}
n\\
k
\end{pmatrix}
\widehat{P}^{n-k}\widehat{X}^{m}\widehat{P}^{k}. \label{w3}%
\end{equation}

It turns out that the Weyl quantization rule (\ref{w3}) for monomials is a
particular case of the so-called \textquotedblleft$\tau$%
-ordering\textquotedblright: for any real number $\tau$ one defines%

\begin{equation}
x^{m}p^{n}\overset{\tau}{\longleftrightarrow}\sum_{k=0}^{n}%
\begin{pmatrix}
n\\
k
\end{pmatrix}
(1-\tau)^{k}\tau^{n-k}\widehat{P}^{k}\widehat{X}^{n}\widehat{P}^{n-k}
\label{taumn}%
\end{equation}
this rule reduces to Weyl's prescription when $\tau=\frac{1}{2}$. When
$\tau=1$ one gets the \textquotedblleft normal ordering\textquotedblright%
\ $\widehat{X}^{n}\widehat{P}^{n}$ familiar from the elementary theory of
partial differential equations, and $\tau=0$ yields the \textquotedblleft
anti-normal ordering\textquotedblright\ $\widehat{P}^{n}\widehat{X}^{n}$
sometimes used in physics. We now make the following fundamental observation:
the Born--Jordan prescription (\ref{bj1}) is obtained by averaging the $\tau
$-ordering on the interval $[0,1]$ (de Gosson \cite{transam}, de Gosson and
Luef \cite{golu1}). In fact
\[
\int_{0}^{1}(1-\tau)^{k}\tau^{n-k}d\tau=\frac{k!(n-k)!}{(n+1)!}%
\]
and hence%
\begin{equation}
x^{m}p^{n}\overset{\text{\textrm{BJ}}}{\longrightarrow}\frac{1}{n+1}\sum
_{k=0}^{m}\widehat{P}^{k}\widehat{X}^{n}\widehat{P}^{n-k}. \label{xpbj}%
\end{equation}

One interesting feature of the quantization rules above is the following:
suppose that the operators $\widehat{X}$ and $\widehat{P}$ are such that%
\[
\lbrack\widehat{X},\widehat{P}]=\widehat{X}\widehat{P}-\widehat{P}\widehat
{X}=i\hbar.
\]
then $[\widehat{X}^{m},\widehat{P}^{n}]$ is independent of the choice of
quantization; in fact (see Crehan \cite{cr89} and the references therein):%
\begin{equation}
\lbrack\widehat{X}^{m},\widehat{P}^{n}]=\sum_{k=1}^{\min(m,n)}(i\hbar)^{k}%
\begin{pmatrix}
m\\
k
\end{pmatrix}%
\begin{pmatrix}
n\\
k
\end{pmatrix}
\widehat{P}^{n-k}\widehat{X}^{m-k}. \label{crecomm}%
\end{equation}

In physics as well as in mathematics, the question of a \textquotedblleft
good\textquotedblright\ choice of quantization is more than just academic. For
instance, different choices may lead to different spectral properties. The
following example is due to Crehan \cite{cr89}. Consider the Hamiltonian
function%
\[
H(z)=\tfrac{1}{2}(p^{2}+x^{2})+\lambda(p^{2}+x^{2})^{3}.
\]
The term that gives an ordering problem is evidently $(p^{2}+x^{2})^{3}$;
Crehan then shows that the most general quantization invariant under the
symplectic transformation $(x,p)\longmapsto(p,-x)$ is
\[
\widehat{H}=\frac{1}{2}(\widehat{P}^{2}+\widehat{X}^{2})+\lambda(\widehat
{P}^{2}+\widehat{X}^{2})^{3}+\lambda(3\alpha\hbar^{2}-4)(\widehat{P}%
^{2}+\widehat{X}^{2}).
\]
The eigenfunctions of $\widehat{H}$ are those of the harmonic oscillator, and
the corresponding eigenvalues are the numbers%
\[
E_{N}=(N+\tfrac{1}{2})\hbar+\lambda\hbar(2N+1)^{3}+\lambda\hbar(2N+1)(3\alpha
\hbar^{2}-4)
\]
($N=0,1,2,...$) which clearly shows the dependence of the spectrum on the
parameters $\alpha$ and $\lambda$, and hence of the chosen quantization.

\section{Born--Jordan Quantization}

\subsection{First definition}

Let $z_{0}=(x_{0},p_{0})$ and consider the \textquotedblleft
displacement\textquotedblright\ Hamiltonian function $H_{z_{0}}=\sigma
(z,z_{0})$. The flow determined by the corresponding Hamilton equations is
given by $f_{t}(z)=z+tz_{0}$; for $\Psi\in\mathcal{S}^{\prime}(\mathbb{R}%
^{2n})$ we define $T(z_{0})\Psi(z)=(f_{1})^{\ast}\Psi(z)=\Psi(z-z_{0})$. The
$\tau$-quantization of $H_{z_{0}}$ is the operator $\widehat{H}_{z_{0}}%
=\sigma(\widehat{Z},z_{0}),$ $\widehat{Z}=(\widehat{X},\widehat{P})$; the
solution of the corresponding Schr\"{o}dinger equation at time $t=1$ with
initial condition $\psi$ is given by the Heisenberg operator $\widehat
{T}(z_{0})=e^{\frac{i}{\hbar}\sigma(\widehat{Z},z_{0})}$; its action on
$\psi\in\mathcal{S}^{\prime}(\mathbb{R}^{n})$ is explicitly given by
\begin{equation}
\widehat{T}(z_{0})\psi(x)=e^{\frac{i}{\hbar}(p_{0}x_{0}-\frac{1}{2}p_{0}%
x_{0})}\psi(x-x_{0}). \label{hw}%
\end{equation}
Let $a\in\mathcal{S}^{\prime}(\mathbb{R}^{2n})$ be an observable (or
\textquotedblleft symbol\textquotedblright). By definition, the Weyl
correspondence $a\overset{\text{\textrm{Weyl}}}{\longleftrightarrow
}A_{\mathrm{W}}$ is defined by%
\begin{equation}
A_{\mathrm{W}}\psi=\left(  \tfrac{1}{2\pi\hbar}\right)  ^{n}\int
_{\mathbb{R}^{2n}}a_{\sigma}(z)\widehat{T}(z)\psi dz \label{ahat}%
\end{equation}
where $a_{\sigma}=\mathcal{F}_{\sigma}a$ is the symplectic Fourier transform
of $a$, that is
\begin{equation}
a_{\sigma}(z)=\left(  \tfrac{1}{2\pi\hbar}\right)  ^{n}\langle e^{-\frac
{i}{\hbar}\sigma(z,\cdot)},\rangle; \label{asigmadis}%
\end{equation}
for $a\in\mathcal{S}(\mathbb{R}^{2n})$; informally
\begin{equation}
a_{\sigma}(z)=\left(  \tfrac{1}{2\pi\hbar}\right)  ^{n}\int_{\mathbb{R}^{2n}%
}e^{-\frac{i}{\hbar}\sigma(z,z^{\prime})}a(z^{\prime})dz^{\prime}.
\label{asig}%
\end{equation}
The symplectic Fourier transform $\mathcal{F}_{\sigma}$ is an involution
($\mathcal{F}_{\sigma}^{2}=I_{\mathrm{d}}$); it is related to the usual
Fourier transform $\mathcal{F}$ on $\mathbb{R}^{2n}$ by the formula
$\mathcal{F}_{\sigma}a(z)=\mathcal{F}a(Jz)$, $J$ the standard symplectic
matrix. The action of the operator $A_{\mathrm{W}}$ on a function $\psi
\in\mathcal{S}^{\prime}(\mathbb{R}^{n})$ is given by
\begin{equation}
A_{\mathrm{W}}\psi=\left(  \tfrac{1}{2\pi\hbar}\right)  ^{n}\int
_{\mathbb{R}^{2n}}a_{\sigma}(z_{0})\widehat{T}(z_{0})\psi dz_{0}. \label{apsi}%
\end{equation}

How can we modify this formula to define Born--Jordan quantization? An
apparently easy answer would be to first define $\tau$-quantization by
replacing $\widehat{H}_{z_{0}}$ by its $\tau$-quantized version $\widehat
{H}_{z_{0},\tau}$, and then to average the associated operators $\widehat
{T}_{\tau}(z_{0})$ thus obtained to get a \textquotedblleft$\widehat
{T}_{\mathrm{BJ}}(z_{0})$ operator\textquotedblright\ which would allow to
define $\widehat{A}_{\mathrm{BJ}}$. However, such a procedure trivially fails,
because all $\tau$-quantizations of the displacement Hamiltonian $H_{z_{0}}$
coincide with $\widehat{H}_{z_{0}}$ as can be verified using the polynomial
rule (\ref{taumn}). There is however a simple way out of this difficulty; it
consists in replacing, as we did in \cite{golu1}, $\widehat{T}(z_{0})$ by
$\Theta(z_{0})\widehat{T}(z_{0})$ where
\begin{equation}
\Theta(z_{0})=\frac{\sin(p_{0}x_{0}/2\hbar)}{p_{0}x_{0}/2\hbar} \label{tbj1}%
\end{equation}
and we define the Born--Jordan operator $A_{\mathrm{BJ}}$ by
\begin{equation}
A_{\mathrm{BJ}}\psi=\left(  \tfrac{1}{2\pi\hbar}\right)  ^{n}\int
_{\mathbb{R}^{2n}}a_{\sigma}(z)\Theta(z)\widehat{T}(z_{0})\psi dz.
\label{abj0}%
\end{equation}
This formula will be justified below.

\subsection{Pseudo-differential formulation}

There is another way to describe Born--Jordan quantization. Writing formula
(\ref{apsi}) in pseudo-differential form yields the usual formal expression%
\begin{equation}
A_{\mathrm{W}}\psi(x)=\left(  \tfrac{1}{2\pi\hbar}\right)  ^{n}\int
_{\mathbb{R}^{2n}}e^{\frac{i}{\hbar}p(x-y)}a(\tfrac{1}{2}(x+y),p)\psi(y)dpdy
\label{ahat1}%
\end{equation}
for the Weyl correspondence (we assume for simplicity that $a\in
\mathcal{S}(\mathbb{R}^{2n})$ and $\psi\in\mathcal{S}(\mathbb{R}^{n})$). We
now define the $\tau$-dependent operator \emph{\`{a} la }Shubin\emph{
}\cite{sh87}:%
\begin{equation}
A_{\tau}\psi(x)=\left(  \tfrac{1}{2\pi\hbar}\right)  ^{n}\int_{\mathbb{R}%
^{2n}}e^{\frac{i}{\hbar}p(x-y)}a(\tau x+(1-\tau)y),p)\psi(y)dpdy;
\label{atau1}%
\end{equation}
the Born--Jordan operator $A_{\mathrm{BJ}}$ with symbol $a$ is then defined by
the average
\begin{equation}
A_{\mathrm{BJ}}\psi=\int_{0}^{1}A_{\tau}\psi d\tau\label{abj}%
\end{equation}
which we can write, interchanging the order of the integrations,%
\begin{equation}
A_{\mathrm{BJ}}\psi(x)=\left(  \tfrac{1}{2\pi\hbar}\right)  ^{n}%
\int_{\mathbb{R}^{2n}}e^{\frac{i}{\hbar}p(x-y)}a_{\mathrm{BJ}}(x,y,p)\psi
(y)dpdy \label{abj1}%
\end{equation}
where%
\begin{equation}
a_{\mathrm{BJ}}(x,y,p)=\int_{0}^{1}a(\tau x+(1-\tau)y,p)d\tau. \label{abjsymb}%
\end{equation}
We have been a little bit sloppy in writing the (usually divergent) integrals
above, but all three definitions become rigorous if we view the operators
$A_{\mathrm{W}}$, $A_{\tau}$, and $A_{\mathrm{BJ}}$ as being defined by the
distributional kernels%
\begin{equation}
K(x,y)=\left(  \tfrac{1}{2\pi\hbar}\right)  ^{n/2}(\mathcal{F}_{2}%
^{-1}a)(\tfrac{1}{2}(x+y),p) \label{ker}%
\end{equation}%
\begin{equation}
K_{\tau}(x,y)=\left(  \tfrac{1}{2\pi\hbar}\right)  ^{n/2}(\mathcal{F}_{2}%
^{-1}a)((\tau x+1-\tau)y,p) \label{kertau}%
\end{equation}
($\mathcal{F}_{2}^{-1}$ is the inverse partial Fourier transform with respect
to the second set of variables) and%
\begin{equation}
K_{\mathrm{BJ}}(x,y)=\int_{0}^{1}K_{\tau}(x,y)d\tau\label{kerbj}%
\end{equation}
where $\mathcal{F}_{2}^{-1}$ is the inverse Fourier transform in the second
set of variables. We will give below an alternative rigorous definition, but
let us first check that definition (\ref{abj1})--(\ref{kerbj}) coincides with
the one given in previous subsection. Define the modified Heisenberg--Weyl
operators
\begin{equation}
\widehat{T}_{\tau}(z_{0})\psi(x)=e^{\frac{i}{2\hbar}(2\tau-1)p_{0}x_{0}%
}\widehat{T}(z_{0})\psi(x) \label{hopt}%
\end{equation}
that is
\begin{equation}
\widehat{T}_{\tau}(z_{0})\psi(x)=e^{\frac{i}{\hbar}(p_{0}x-(1-\tau)p_{0}%
x_{0})}\psi(x-x_{0}). \label{hoptbis}%
\end{equation}
These obey the same commutation rules%
\begin{equation}
\widehat{T}_{\tau}(z_{0})\widehat{T}_{\tau}(z_{1})=e^{\frac{i}{\hbar}%
\sigma(z_{0},z_{1})}\widehat{T}_{\tau}(z_{1})\widehat{T}_{\tau}(z_{0})
\label{comm}%
\end{equation}
as the usual Heisenberg operators $\widehat{T}(z_{0})$.

\begin{proposition}
\label{propobj}Let $a\in\mathcal{S}^{\prime}(\mathbb{R}^{2n})$, $\psi
\in\mathcal{S}(\mathbb{R}^{n})$. The Born--Jordan operator (\ref{abj}) is
given by formula (\ref{abj0}), that is
\begin{equation}
A_{\mathrm{BJ}}\psi=\left(  \tfrac{1}{2\pi\hbar}\right)  ^{n}\int
_{\mathbb{R}^{2n}}a_{\sigma}(z)\widehat{T}_{\mathrm{BJ}}(z)\psi dz
\label{abjobis}%
\end{equation}
with%
\begin{equation}
\widehat{T}_{\mathrm{BJ}}(z)=\Theta(z)\widehat{T}(z)\text{ \ , \ }%
\Theta(z)=\frac{\sin(px/2\hbar)}{px/2\hbar}. \label{tbj}%
\end{equation}
In particular, $A_{\mathrm{BJ}}$ is the Weyl operator with symbol
\begin{equation}
a_{\mathrm{BJ}}=\left(  \tfrac{1}{2\pi\hbar}\right)  ^{n}a\ast\mathcal{F}%
_{\sigma}\Theta\label{b}%
\end{equation}

\end{proposition}

\begin{proof}
(\textit{Cf}. \cite{golu1,transam}). One verifies by a straightforward
computation that the Shubin formula (\ref{atau1}) can be rewritten as%
\begin{equation}
A_{\tau}\psi=\int_{\mathbb{R}^{2n}}a_{\sigma}(z)\widehat{T}_{\tau}(z)\psi dz.
\label{ataupsi}%
\end{equation}
Let us now average in $\tau$ over the interval $[0,1]$; interchanging the
order of integrations and using the trivial identity%
\[
\int_{0}^{1}e^{\frac{i}{2\hbar}(2\tau-1)px}d\tau=\frac{2\hbar}{px}\sin
\frac{px}{2\hbar}%
\]
we get
\[
\widehat{T}_{\mathrm{BJ}}(z)=\int_{0}^{1}\widehat{T}_{\tau}(z)d\tau
=\Theta(z)\widehat{T}(z)
\]
hence formula (\ref{abjobis}). To prove the last statement we note that
formula (\ref{abjobis}) can be rewritten%
\begin{align*}
A_{\mathrm{BJ}}\psi &  =\left(  \tfrac{1}{2\pi\hbar}\right)  ^{n}%
\int_{\mathbb{R}^{2n}}a_{\sigma}(z)\Theta(z)\widehat{T}(z)\psi dz\\
&  =\left(  \tfrac{1}{2\pi\hbar}\right)  ^{n}\int_{\mathbb{R}^{2n}%
}(a_{\mathrm{BJ}})_{\sigma}(z)\widehat{T}(z)\psi dz
\end{align*}
where $(a_{\mathrm{BJ}})_{\sigma}=a\Theta$. Taking the inverse Fourier
transform we get, noting that $\mathcal{F}_{\sigma}(a\ast b)=(2\pi\hbar
)^{n}\mathcal{F}_{\sigma}a\mathcal{F}_{\sigma}b$,%
\[
b=\left(  \tfrac{1}{2\pi\hbar}\right)  ^{n}a\ast\mathcal{F}_{\sigma}%
^{-1}\Theta
\]
hence (\ref{b}) since $\mathcal{F}_{\sigma}^{-1}\Theta=\mathcal{F}_{\sigma
}\Theta$ because the function $\Theta$ is even.
\end{proof}

One easily verifies that the (formal) adjoint of $A_{\tau}=\operatorname*{Op}%
_{\tau}(a)$ is given by
\begin{equation}
\operatorname*{Op}\nolimits_{\tau}(a)^{\ast}=\operatorname*{Op}%
\nolimits_{1-\tau}(\overline{a}) \label{opadj}%
\end{equation}
and hence
\begin{equation}
\operatorname*{Op}\nolimits_{\mathrm{BJ}}(a)^{\ast}=\operatorname*{Op}%
\nolimits_{\mathrm{BJ}}(\overline{a}). \label{bjadj}%
\end{equation}
Born--Jordan operators thus share with Weyl operators the property of being
(essentially) self-adjoint if and only if their symbol is real. This property
makes Born--Jordan prescription a good candidate for physical quantization,
while Shubin quantization should be rejected being unphysical for $\tau
\neq\frac{1}{2}$.

\section{The Born--Jordan--Wigner Distribution}

\subsection{The $\tau$-Wigner distribution}

In a recent series of papers Boggiatto and his collaborators
\cite{bogetal,bogetalbis,bogetalter} have introduced a $\tau$-dependent Wigner
distribution $\operatorname*{Wig}\nolimits_{\tau}(f,g)$ which they average
over the values of $\tau$ in the interval $[0,1]$. This procedure leads to an
element of the Cohen class \cite{Gro}, \textit{i.e.} to a transform of the
type $C(\psi,\phi)=\operatorname*{Wig}\nolimits_{\tau}(\psi,\phi)\ast\theta$
where $\theta\in\mathcal{S}^{\prime}(\mathbb{R}^{2n})$. From the point of view
of time-frequency analysis this can be interpreted as the application of a
filter to the Wigner transform.

Let us define the $\tau$-Wigner cross-distribution $\operatorname*{Wig}%
\nolimits_{\tau}(\psi,\phi)$ of a pair $(\psi,\phi)$ of functions in
$\mathcal{S}(\mathbb{R}^{n})$:
\begin{equation}
\operatorname*{Wig}\nolimits_{\tau}(\psi,\phi)(z)=\left(  \tfrac{1}{2\pi\hbar
}\right)  ^{n}\int_{\mathbb{R}^{n}}e^{-\frac{i}{\hbar}py}\psi(x+\tau
y)\overline{\phi}(x-(1-\tau)y)dy. \label{wtau}%
\end{equation}
Choosing $\tau=\frac{1}{2}$ one recovers the usual cross-Wigner transform
\begin{equation}
\operatorname*{Wig}(\psi,\phi)(z)=\left(  \tfrac{1}{2\pi\hbar}\right)
^{n}\int_{\mathbb{R}^{n}}e^{-\frac{i}{\hbar}py}\psi(x+\tfrac{1}{2}%
y)\overline{\phi}(x-\tfrac{1}{2}y)dy \label{wigo1}%
\end{equation}
and when $\tau=0$ we get the Rihaczek--Kirkwood distribution%
\begin{equation}
R(\psi,\phi)(z)=\left(  \tfrac{1}{2\pi\hbar}\right)  ^{n/2}e^{-\frac{i}{\hbar
}px}\psi(x)\overline{\mathcal{F}\phi}(p) \label{riki}%
\end{equation}
well-known from time-frequency analysis \cite{Gro}.

The mapping $\operatorname*{Wig}\nolimits_{\tau}$ is a bilinear and continuous
mapping $\mathcal{S}(\mathbb{R}^{n})\times\mathcal{S}(\mathbb{R}%
^{n})\longrightarrow\mathcal{S}(\mathbb{R}^{2n})$. When $\psi=\phi$ one writes
$\operatorname*{Wig}\nolimits_{\tau}(\psi,\psi)=\operatorname*{Wig}%
\nolimits_{\tau}\psi$; it is the $\tau$-Wigner distribution considered by
Boggiatto \textit{et al}. \cite{bogetal,bogetalbis,bogetalter}). It follows
from the definition of $\operatorname*{Wig}\nolimits_{\tau}$ that we have
\begin{equation}
\overline{\operatorname*{Wig}\nolimits_{\tau}(\phi,\psi)}=\operatorname*{Wig}%
\nolimits_{1-\tau}(\psi,\phi); \label{wtauconj1}%
\end{equation}
in particular%
\begin{equation}
\overline{\operatorname*{Wig}\nolimits_{\tau}\psi}=\operatorname*{Wig}%
\nolimits_{1-\tau}\psi\label{wtauconj2}%
\end{equation}
hence $\operatorname*{Wig}\nolimits_{\tau}\psi$ is not a real function in
general if $\tau\neq\frac{1}{2}$.

\begin{proposition}
Assume that $\psi,\phi\in L^{1}(\mathbb{R}^{n})\cap L^{2}(\mathbb{R}^{n})$.
Then%
\begin{equation}
\int_{\mathbb{R}^{n}}\operatorname*{Wig}\nolimits_{\tau}(\psi,\phi
)(z)dp=\psi(x)\overline{\phi}(x) \label{cromarg1}%
\end{equation}
and%
\begin{equation}
\int_{\mathbb{R}^{n}}\operatorname*{Wig}\nolimits_{\tau}(\psi,\phi
)(z)dx=\mathcal{F}\psi(p)\overline{\mathcal{F}\phi(p)}. \label{cromarg2}%
\end{equation}

\end{proposition}

\begin{proof}
Formula (\ref{cromarg1}) is straightforward. On the other hand%
\[
\int_{\mathbb{R}^{n}}\operatorname*{Wig}\nolimits_{\tau}(\psi,\phi
)(z)dx=\left(  \tfrac{1}{2\pi\hbar}\right)  ^{n}\int_{\mathbb{R}^{2n}%
}e^{-\frac{i}{\hbar}py}\psi(x+\tau y)\overline{\phi}(x-(1-\tau)y)dxdy
\]
and setting $x^{\prime}=x+\tau y$, $x^{\prime\prime}=x-(1-\tau)y$ we have
$dx^{\prime}dx^{\prime\prime}=dxdy$ so that%
\[
\int_{\mathbb{R}^{n}}\operatorname*{Wig}\nolimits_{\tau}(\psi,\phi
)(z)dx=\left(  \tfrac{1}{2\pi\hbar}\right)  ^{n}\int_{\mathbb{R}^{2n}%
}e^{-\frac{i}{\hbar}px^{\prime}}\psi(x^{\prime})e^{\frac{i}{\hbar}%
px^{\prime\prime}}\overline{\phi}(x^{\prime\prime})dxdy
\]
hence formula (\ref{cromarg2}). Notice that the right-hand sides of
(\ref{cromarg1}) and (\ref{cromarg2}) are independent of the parameter $\tau$.
\end{proof}

In particular \cite{bogetal}, the $\tau$-Wigner distribution
$\operatorname*{Wig}\nolimits_{\tau}\psi=\operatorname*{Wig}\nolimits_{\tau
}(\psi,\psi)$ satisfies the usual marginal properties:%
\begin{equation}
\int_{\mathbb{R}^{n}}\operatorname*{Wig}\nolimits_{\tau}\psi(z)dp=|\psi
(x)|^{2}\text{ ,\ }\int_{\mathbb{R}^{n}}\operatorname*{Wig}\nolimits_{\tau
}\psi(z)dx=|\mathcal{F}\psi(p)|^{2}.\text{\ } \label{marginal1}%
\end{equation}

There is a fundamental relation between Weyl pseudo-differential operators and
the cross-Wigner transform, that relation is often used to define the Weyl
operator $A_{\mathrm{W}}=\operatorname*{Op}_{\mathrm{W}}(a)$:
\begin{equation}
\langle A_{\mathrm{W}}\psi|\overline{\phi}\rangle=\langle
a,\operatorname*{Wig}(\psi,\phi)\rangle\label{aweyl}%
\end{equation}
for $\psi,\phi\in\mathcal{S}(\mathbb{R}^{n})$. Not very surprisingly this
formula extends to the case of $\tau$-operators:

\begin{proposition}
\label{propos1}Let $\psi,\phi\in\mathcal{S}(\mathbb{R}^{n})$, $a\in
\mathcal{S}(\mathbb{R}^{2n})$, and $\tau$ a real number. We have%
\begin{equation}
\langle A_{\tau}\psi|\phi\rangle=\langle a,\operatorname*{Wig}\nolimits_{\tau
}(\psi,\phi)\rangle\label{awtau}%
\end{equation}
where $\langle\cdot,\cdot\rangle$ is the distributional bracket on
$\mathbb{R}^{2n}$ and $A_{\tau}=\operatorname*{Op}\nolimits_{\tau}(a)$.
\end{proposition}

\begin{proof}
By definition of $\operatorname*{Wig}\nolimits_{\tau}$ we have%
\begin{multline*}
\langle a,\operatorname*{Wig}\nolimits_{\tau}(\psi,\phi)\rangle=\\
\left(  \tfrac{1}{2\pi\hbar}\right)  ^{n}\int_{\mathbb{R}^{3n}}e^{-\frac
{i}{\hbar}py}a(z)\psi(x+\tau y)\overline{\phi}(x-(1-\tau)y)dydpdx.
\end{multline*}
Defining new variables $x^{\prime}=x-(1-\tau)y$ and $y^{\prime}=x+\tau y$ we
have $y=y^{\prime}-x^{\prime}$, $dydx=dy^{\prime}dx^{\prime}$ and hence
\begin{multline*}
\langle a,\operatorname*{Wig}\nolimits_{\tau}(\psi,\phi)\rangle=\\
\left(  \tfrac{1}{2\pi\hbar}\right)  ^{n}\int_{\mathbb{R}^{3n}}e^{-\frac
{i}{\hbar}p(x^{\prime}-y^{\prime})}a(\tau x^{\prime}+(1-\tau)y^{\prime}%
,p)\psi(y^{\prime})\overline{\phi}(x^{\prime})dy^{\prime}dp^{\prime}%
dx^{\prime};
\end{multline*}
the equality (\ref{awtau}) follows in view of definition (\ref{atau1}) of
$A_{\tau}$.
\end{proof}

Formula (\ref{awtau}) allows us to define $\widehat{A}_{\tau}\psi
=\operatorname*{Op}\nolimits_{\tau}(a)\psi$ for arbitrary symbols
$a\in\mathcal{S}^{\prime}(\mathbb{R}^{2n})$ and $\psi\in\mathcal{S}%
(\mathbb{R}^{n})$ in the same way as is done for Weyl pseudo-differential
operators: choose $\phi\in\mathcal{S}(\mathbb{R}^{n})$; then
$\operatorname*{Wig}\nolimits_{\tau}(\psi,\phi)\in\mathcal{S}(\mathbb{R}%
^{2n})$ and the distributional bracket $\langle a,\operatorname*{Wig}%
\nolimits_{\tau}(\psi,\phi)\rangle$ is thus well-defined. This defines
$\widehat{A}_{\tau}$ as a continuous operator $\mathcal{S}(\mathbb{R}%
^{n})\longrightarrow\mathcal{S}^{\prime}(\mathbb{R}^{n})$.

\subsection{Averaging over $\tau$}

We define the (cross) Born--Jordan--Wigner (BJW) distribution of $\psi,\phi
\in\mathcal{S}(\mathbb{R}^{n})$ by the formula%
\begin{equation}
\operatorname*{Wig}\nolimits_{\mathrm{BJ}}(\psi,\phi)(z)=\int_{0}%
^{1}\operatorname*{Wig}\nolimits_{\tau}(\psi,\phi)d\tau. \label{wbj}%
\end{equation}
We set $\operatorname*{Wig}\nolimits_{\mathrm{BJ}}\psi=\operatorname*{Wig}%
\nolimits_{\mathrm{BJ}}(\psi,\psi)$. The properties of the BJW distribution
are readily deduced from those of the $\tau$-Wigner distribution studied
above. In particular, the marginal properties (\ref{cromarg1}) and
(\ref{cromarg2}) are obviously preserved:%
\begin{equation}
\int_{\mathbb{R}^{n}}\operatorname*{Wig}\nolimits_{\mathrm{BJ}}(\psi
,\phi)(z)dp=\psi(x)\overline{\phi}(x); \label{bjcromarg1}%
\end{equation}
and
\begin{equation}
\int_{\mathbb{R}^{n}}\operatorname*{Wig}\nolimits_{\mathrm{BJ}}(\psi
,\phi)(z)dx=\mathcal{F}\psi(p)\overline{\mathcal{F}\phi(p)}.
\label{bjcromarg2}%
\end{equation}
An object closely related to the (cross-)Wigner distribution is the
(cross-)ambiguity function of a pair of functions $\psi,\phi\in\mathcal{S}%
(\mathbb{R}^{n})$:%
\[
A(\psi,\phi)(z)=\left(  \tfrac{1}{2\pi\hbar}\right)  ^{n}\int_{\mathbb{R}^{n}%
}e^{-\tfrac{i}{\hbar}px^{\prime}}\psi(x^{\prime}+\tfrac{1}{2}x)\overline{\phi
}(x^{\prime}-\tfrac{1}{2}x)dx^{\prime}.
\]
It turns out that $A(\psi,\phi)$ and $\operatorname*{Wig}\nolimits_{\tau}%
(\psi,\phi)$ are obtained from each other by a symplectic Fourier transform:
\[
\operatorname*{Amb}(\psi,\phi)=\mathcal{F}_{\sigma}\operatorname*{Wig}%
(\psi,\phi)\text{ , }\operatorname*{Wig}(\psi,\phi)=\mathcal{F}_{\sigma
}\operatorname*{Amb}(\psi,\phi)
\]
($\mathcal{F}_{\sigma}$ is an involution), thus justifying the following
definition:
\[
\operatorname*{Amb}\nolimits_{\mathrm{BJ}}(\psi,\phi)=\mathcal{F}_{\sigma
}\operatorname*{Amb}(\psi,\phi).
\]

An important property is that the BJW distribution of a function $\psi$ is
real, as is the usual Wigner distribution. In fact using the conjugacy formula
(\ref{wtauconj2}) we have
\begin{align*}
\overline{\operatorname*{Wig}\nolimits_{\mathrm{BJ}}\psi}  &  =\int_{0}%
^{1}\overline{\operatorname*{Wig}\nolimits_{\tau}\psi}d\tau=\int_{0}%
^{1}\operatorname*{Wig}\nolimits_{1-\tau}\psi d\tau\\
&  =\int_{0}^{1}\operatorname*{Wig}\nolimits_{\tau}\psi d\tau
=\operatorname*{Wig}\nolimits_{\mathrm{BJ}}\psi.
\end{align*}

Formula (\ref{awtau}) relating Shubin's $\tau$-operators to the $\tau$-(cross)
Wigner distribution carries over to the Born--Jordan case:

\begin{corollary}
Let $a\in\mathcal{S}^{\prime}(\mathbb{R}^{2n})$, $A_{\mathrm{BJ}%
}=\operatorname*{Op}_{\mathrm{BJ}}(a)$. We have%
\begin{equation}
\langle A_{\mathrm{BJ}}\psi|\overline{\phi}\rangle=\langle
a,\operatorname*{Wig}\nolimits_{\mathrm{BJ}}(\psi,\phi)\rangle\label{opjaw}%
\end{equation}
for all $\psi,\phi\in\mathcal{S}(\mathbb{R}^{n})$.
\end{corollary}

\begin{proof}
It suffices to integrate the equality $\langle A_{\tau}\psi|\phi
\rangle=\langle a,\operatorname*{Wig}\nolimits_{\tau}(\psi,\phi)\rangle$ with
respect to $\tau\in\lbrack0,1]$ and to use definitions (\ref{abj}) and
(\ref{wbj}).
\end{proof}

One defines $A_{\mathrm{BJ}}=\operatorname*{Op}\nolimits_{\mathrm{BJ}}(a)$ for
arbitrary $a\in\mathcal{S}^{\prime}(\mathbb{R}^{2n})$ by the same procedure as
for Weyl operators and noting that $\operatorname*{Wig}\nolimits_{\mathrm{BJ}%
}(\psi,\phi)\in\mathcal{S}(\mathbb{R}^{2n})$ if $\psi,\phi\in\mathcal{S}%
(\mathbb{R}^{n})$.

The following consequence of Proposition \ref{propos1} will be essential in
our study of the uncertainty principle:

\begin{proposition}
Let $\Theta$ be defined by (\ref{tbj}). We have%
\begin{equation}
\operatorname*{Wig}\nolimits_{\mathrm{BJ}}(\psi,\phi)=\operatorname*{Wig}%
(\psi,\phi)\ast\mathcal{F}\Theta\label{wbjthe}%
\end{equation}
and%
\begin{equation}
\operatorname*{Amb}\nolimits_{\mathrm{BJ}}(\psi,\phi)=(2\pi\hbar
)^{n}\operatorname*{Amb}(\psi,\phi)\Theta\label{abjthe}%
\end{equation}

\end{proposition}

\begin{proof}
In view of formula (\ref{b}) in Proposition \ref{propobj} and formula
(\ref{wbj}) above we have%
\[
\langle a,\operatorname*{Wig}\nolimits_{\mathrm{BJ}}(\psi,\phi)\rangle=\langle
b,\operatorname*{Wig}(\psi,\phi)\rangle
\]
where $b=(2\pi\hbar)^{n}a\ast\mathcal{F}_{\sigma}\Theta$, hence%
\begin{align*}
\langle a,\operatorname*{Wig}\nolimits_{\mathrm{BJ}}(\psi,\phi)\rangle &
=(2\pi\hbar)^{n}\langle a\ast\mathcal{F}_{\sigma}\Theta,\operatorname*{Wig}%
(\psi,\phi)\rangle\\
&  =(2\pi\hbar)^{n}\langle a,[\mathcal{F}_{\sigma}\Theta\circ(-I_{\mathrm{d}%
})]\ast\operatorname*{Wig}(\psi,\phi)\rangle\\
&  =(2\pi\hbar)^{n}\langle a,\mathcal{F}_{\sigma}\Theta\ast\operatorname*{Wig}%
(\psi,\phi)\rangle
\end{align*}
(the last equality because $\Theta$ is an even function). Since $\mathcal{F}%
_{\sigma}\Theta=\mathcal{F}\Theta$ because $\Theta$ is even and invariant
under permutation of the $x$ and $p$ variables, this proves formula
(\ref{wbjthe}), and formula (\ref{abjthe}) follows, taking the symplectic
Fourier transform of both sides.
\end{proof}

\subsection{Symplectic (non-)covariance}

he symplectic group $\operatorname*{Sp}(2n,\mathbb{R})$ is by definition the
group of all linear automorphisms $s$ of $\mathbb{R}^{2n}$ which preserve the
symplectic form $\sigma(z,z^{\prime})=$ $Jz\cdot z^{\prime}$; equivalently
$s^{T}Js=J$. The group $\operatorname*{Sp}(2n,\mathbb{R})$ is a connected Lie
group, and its double covering $\operatorname*{Sp}_{2}(2n,\mathbb{R})$ has a
faithful (but reducible) representation by a group of unitary operator, the
metaplectic group $\operatorname*{Mp}(2n,\mathbb{R})$ (see Folland
\cite{Folland}, de Gosson \cite{Birk,Birkbis}). That group is generated by the
operators $\widehat{J}$, $\widehat{M}_{L,m}$, and $\widehat{V}_{-P}$ defined
by $\widehat{J}=e^{-in\pi/4}\mathcal{F}$ and
\[
\widehat{M}_{L,m}\psi(x)=i^{m}\sqrt{|\det L|}\psi(Lx)\text{ , }\widehat{V}%
_{P}=e^{-\frac{i}{2\hbar}Px\cdot x}\psi(x)
\]
where $L\in G\ell(n,\mathbb{R})$ and $P\in Sym(n,\mathbb{R})$; the integer $m$
corresponds to the choice of an argument of $\det L$. Denoting by
$\pi^{\operatorname*{Mp}}$ the covering projection $\operatorname*{Mp}%
(2n,\mathbb{R})\longrightarrow\operatorname*{Sp}(2n,\mathbb{R})$ we have
$\pi^{\operatorname*{Mp}}(\widehat{J})=J$ and%
\[
\pi^{\operatorname*{Mp}}(\widehat{M}_{L,m})=%
\begin{pmatrix}
L^{-1} & 0\\
0 & L^{T}%
\end{pmatrix}
\text{ , }\pi^{\operatorname*{Mp}}(\widehat{V}_{P})=%
\begin{pmatrix}
I & 0\\
P & I
\end{pmatrix}
.
\]

Let now $A_{\mathrm{W}}=\operatorname*{Op}_{\mathrm{W}}(a)$ be an arbitrary
Weyl operator, and $s\in\operatorname*{Sp}(2n,\mathbb{R})$. We have%
\begin{equation}
\widehat{S}A_{\mathrm{W}}\widehat{S}^{-1}=\operatorname*{Op}%
\nolimits_{\mathrm{W}}(a\circ s^{-1}) \label{metaco1}%
\end{equation}
where $\widehat{S}\in\operatorname*{Mp}(2n,\mathbb{R})$ is anyone of the two
metaplectic operators such that $\pi^{\operatorname*{Mp}}(\widehat{S})=s$.
This property is really characteristic of the Weyl correspondence; it is
proven \cite{Birk,Birkbis} using the identity
\begin{equation}
\widehat{S}\widehat{T}(z)=\widehat{T}(sz)\widehat{S} \label{stsz}%
\end{equation}
where $\widehat{T}(z)$ is the Heisenberg operator. One can show that if a
pseudo-differential correspondence $a\longleftrightarrow\operatorname*{Op}(a)$
(admissible in the sense above, or not) is such that $\operatorname*{Op}%
(a\circ s^{-1})=\widehat{S}\operatorname*{Op}(a)\widehat{S}^{-1}$ then it
\emph{must} be the Weyl correspondence. For Born--Jordan correspondence we do
still have a residual symplectic covariance, namely:

\begin{proposition}
\label{propams}Let $A_{\text{\textrm{BJ}}}=\operatorname*{Op}%
\nolimits_{\mathrm{BJ}}(a)$ with $a\in\mathcal{S}^{\prime}(\mathbb{R}^{2n})$.
We have%
\begin{equation}
\widehat{S}\operatorname*{Op}\nolimits_{\mathrm{BJ}}(a)\widehat{S}%
^{-1}=\operatorname*{Op}\nolimits_{\mathrm{BJ}}(a\circ s^{-1}) \label{cobj1}%
\end{equation}
for every $\widehat{S}$ in the subgroup of $\operatorname*{Mp}(2n,\mathbb{R})$
generated by the operators $\widehat{J}$ and $\widehat{M}_{L,m}$ (with
$\pi^{\operatorname*{Mp}}(\widehat{S})=s$).
\end{proposition}

\begin{proof}
(Cf. de Gosson \cite{transam}). It suffices to prove formula (\ref{cobj1}) for
$\widehat{S}=\widehat{J}$ and $\widehat{S}=\widehat{M}_{L,m}$. Let first
$\widehat{S}$ be an arbitrary element of $\operatorname*{Mp}(2n,\mathbb{R})$;
we have
\begin{align*}
\widehat{S}\operatorname*{Op}\nolimits_{\mathrm{BJ}}(a)  &  =\left(  \tfrac
{1}{2\pi\hbar}\right)  ^{n}\int a_{\sigma}(z)\Theta(z)\widehat{S}\widehat
{T}(z)dz\\
&  =\left(  \tfrac{1}{2\pi\hbar}\right)  ^{n}\left(  \int a_{\sigma}%
(z)\Theta(z)\widehat{T}(sz)dz\right)  \widehat{S}%
\end{align*}
where the second equality follows from the symplectic covariance property
(\ref{stsz}) of the Heisenberg operators. Making the change of variables
$z^{\prime}=sz$ in the integral we get, since $\det s=1$,%
\[
\int a_{\sigma}(z)\Theta(z)\widehat{T}(sz)dz=\int a_{\sigma}(s^{-1}%
z)\Theta(s^{-1}z)\widehat{T}(z)dz.
\]
Now, by definition of the symplectic Fourier transform we have%
\[
a_{\sigma}(s^{-1}z)=\left(  \tfrac{1}{2\pi\hbar}\right)  ^{n}\int e^{-\frac
{i}{\hbar}\sigma(s^{-1}z,z^{\prime})}a(z^{\prime})dz^{\prime}=(a\circ
s^{-1})_{\sigma}(z).
\]
Let now $\widehat{S}=\widehat{M}_{L,m}$; we have%
\[
\Theta(M_{L}^{-1}z)=\frac{\sin(Lp(L^{T})^{-1}x/2\hbar)}{Lp(L^{T})^{-1}%
x/2\hbar}=\Theta(z);
\]
similarly $\Theta(J^{-1}z)=\Theta(z)$, hence in both cases%
\begin{align*}
\widehat{S}\operatorname*{Op}\nolimits_{\mathrm{BJ}}(a)  &  =\left(  \tfrac
{1}{2\pi\hbar}\right)  ^{n}\left(  \int(a\circ s^{-1})_{\sigma}\Theta
(z)\widehat{T}(z)dz\right)  \widehat{S}\\
&  =\operatorname*{Op}\nolimits_{\mathrm{BJ}}(a\circ s^{-1})\widehat{S}%
\end{align*}
whence formula (\ref{cobj1}).
\end{proof}

\section{The Uncertainty Principle}

We begin by reviewing the notion of density matrix (or operator) familiar from
statistical quantum mechanics. The notion goes back to John von Neumann
\cite{jvn} in 1927, and is intimately related to the notion of mixed state
(whose study mathematically belongs to the theory of $C^{\ast}$-algebras via
the GNS construction). This will provide us with all the necessary tools for
comparing the uncertainty relations in the Weyl and Born--Jordan case.

\subsection{Density matrices}

A density matrix on a Hilbert space $\mathcal{H}$ is a self-adjoint positive
operator on $\mathcal{H}$ with trace one. In particular, it is a compact
operator. Physically density matrices represent statistical mixtures of pure
states, as explicitly detailed below.

We will need the two following results:

\begin{lemma}
\label{lemma1}A self-adjoint trace class operator $\widehat{\rho}$ on a
Hilbert space $\mathcal{H}$ is a density matrix if and only if there exists a
sequence $(\alpha_{j})_{j\geq1}$ of positive numbers and a sequence of
pairwise orthogonal finite-dimensional subspaces $(\mathcal{H}_{j})_{j\geq1}%
$\ of $\mathcal{H}$ such that%
\[
\widehat{\rho}=\sum_{j\geq1}\alpha_{j}\Pi_{j}\text{ , }\sum_{j\geq1}%
m_{j}\alpha_{j}=1
\]
where $\Pi_{j}$ is the orthogonal projection of $\mathcal{H}$ on
$\mathcal{H}_{j}$ and $m_{j}=\dim\mathcal{H}_{j}$
\end{lemma}

It is a consequence of the spectral decomposition theorem for compact
operators (for a detailed proof see \textit{e.g.} \cite{Birkbis}, \S 13.1).

\begin{lemma}
\label{lemma2}Let $\psi\in L^{2}(\mathbb{R}^{n})$, $\psi\neq0$. The projection
operator $P_{\psi}:$ $L^{2}(\mathbb{R}^{n})\longrightarrow\{\lambda
\psi:\lambda\in\mathbb{C}\}$ has Weyl and Born--Jordan symbols given by,
respectively%
\begin{equation}
\rho_{\mathrm{W}}=(2\pi\hbar)^{n}\operatorname*{Wig}\nolimits_{\mathrm{W}}%
\psi\label{pitw1}%
\end{equation}
and
\begin{equation}
\rho_{\mathrm{BJ}}=(2\pi\hbar)^{n}\operatorname*{Wig}\nolimits_{\mathrm{BJ}%
}\psi. \label{pitbj1}%
\end{equation}
In particular $\rho_{\mathrm{W}}$ and $\rho_{\mathrm{BJ}}$ are real functions.
\end{lemma}

\begin{proof}
We have $P_{\psi}\phi=(\phi|\psi)_{L^{2}}\psi$ hence the kernel of $P_{\psi}$
is $K_{\psi}=\psi\otimes\overline{\psi}$. Using a Fourier transform formula
(\ref{kertau}) implies that the $\tau$-symbol $\rho_{\tau}$ of $P_{\psi}$ is
given by%
\begin{align*}
\rho_{\tau}(z)  &  =\int_{\mathbb{R}^{n}}e^{-\frac{i}{\hbar}py}K_{\psi}(x+\tau
y,x-(1-\tau)y)dy\\
&  =\int_{\mathbb{R}^{n}}e^{-\frac{i}{\hbar}py}\psi(x+\tau y)\overline{\psi
}(x-(1-\tau)y)dy\\
&  =(2\pi\hbar)^{n}\operatorname*{Wig}\nolimits_{\tau}\psi(z).
\end{align*}
Setting $\tau=\frac{1}{2}$ we get formula (\ref{pitw1}). Formula
(\ref{pitbj1}) is obtained by integrating $\rho_{\tau}(z)$ with respect to
$\tau\in\lbrack0,1]$. That $\rho_{\mathrm{W}}$ and $\rho_{\mathrm{BJ}}$ are
real follows from the fact that both the Wigner and the WBJ distribution are real.
\end{proof}

The following result describes density matrices in both the Weyl and
Born--Jordan case in terms of the Wigner formalism:

\begin{proposition}
Let $\widehat{\rho}$ be a density matrix on $L^{2}(\mathbb{R}^{n})$. There
exists an orthonormal system $(\psi_{j})_{j\geq1}$ of $L^{2}(\mathbb{R}^{n})$
and a sequence of non-negative numbers $(\lambda_{j})_{j\geq1}$ such that
$\sum_{j\geq1}\lambda_{j}=1$ and
\begin{equation}
\widehat{\rho}=\operatorname*{Op}\nolimits_{\mathrm{BJ}}\rho_{\mathrm{BJ}%
}=\operatorname*{Op}\nolimits_{\mathrm{W}}\rho_{\mathrm{W}} \label{bjw}%
\end{equation}
the symbols $\rho_{\mathrm{W}}$ and $\rho_{\mathrm{BJ}}$ being given by
\begin{equation}
\rho_{\mathrm{W}}=(2\pi\hbar)^{n}\sum_{j\geq1}\lambda_{j}\operatorname*{Wig}%
\psi_{j} \label{robw1}%
\end{equation}
and%
\begin{equation}
\rho_{\mathrm{BJ}}=(2\pi\hbar)^{n}\sum_{j\geq1}\lambda_{j}\operatorname*{Wig}%
\nolimits_{\mathrm{BJ}}\psi_{j}\text{.} \label{robj1}%
\end{equation}

\end{proposition}

\begin{proof}
Taking $\mathcal{H}=L^{2}(\mathbb{R}^{n})$ in Lemma \ref{lemma1} we can write
$\widehat{\rho}=\sum_{j}\alpha_{j}\Pi_{j}$ where each $\Pi_{j}$ is the
projection operator on a finite dimensional space $\mathcal{H}_{j}\subset
L^{2}(\mathbb{R}^{n})$, and two spaces $\mathcal{H}_{j}$ and $\mathcal{H}%
_{\ell}$ are orthonormal if $j\neq\ell$. For each index $j$ let us choose an
orthonormal basis $\mathcal{B}_{j}=(\psi_{j+1},...,\psi_{j+m_{j}})$ of
$\mathcal{H}_{j}$ ; the union $\mathcal{B}=\cup_{j}\mathcal{B}_{j}$ is then an
orthonormal basis of $\oplus_{j}\mathcal{H}_{j}$, and we have, using Lemma
\ref{lemma1}
\[
\widehat{\rho}=\sum_{j\geq1}\alpha_{j}\left(  \sum_{j+1\leq k\leq j+m_{j}}%
\Pi_{\psi_{k}}\right)
\]
where $\Pi_{\psi_{k}}$ is the orthogonal projection on the ray $\{\lambda
\psi_{j}:\lambda\in\mathbb{C}\}$. Since each index $\alpha_{j}$ is repeated
$m_{j}$ times due to the expression between brackets, this can be rewritten%
\[
\widehat{\rho}=\sum_{j\geq1}m_{j}\alpha_{j}\Pi_{\psi_{j}}=\sum_{j\geq1}%
\lambda_{j}\Pi_{\psi_{j}}%
\]
with the $\lambda_{j}=m_{j}\alpha_{j}$ summing up to one. In view of Lemma
\ref{lemma2} the Weyl (resp. Born--Jordan) symbol of $\Pi_{\psi_{j}}$ is
$(2\pi\hbar)^{n}\operatorname*{Wig}\nolimits_{\mathrm{W}}\psi$ (resp.
$(2\pi\hbar)^{n}\operatorname*{Wig}\nolimits_{\mathrm{BJ}}\psi$) hence the result.
\end{proof}

Notice that the orthonormal bases $\mathcal{B}_{j}$ in the proof can be chosen
arbitrarily; the decompositions (\ref{robw1}) and (\ref{robj1}) are therefore
not unique. (In Physics, one would say that a mixed quantum state can be
written in infinitely many way as a superposition of pure states, a pure state
being a density operator with symbol a Wigner function).

\subsection{A general uncertainty principle}

In what follows the notation
\[
a\longleftrightarrow\widehat{A}=\operatorname*{Op}(a)
\]
is indifferently the Weyl or the Born--Jordan correspondence. Both have the property:

\begin{quote}
\emph{If the symbol }$a$\emph{ is real, then the operator }$\widehat{A}$\emph{
is essentially self-adjoint (in which case we call it an observable). }
\end{quote}

Notice that the Shubin correspondence does not have this property for
$\tau\neq\frac{1}{2}$ since $A_{\tau}^{\ast}=A_{1-\tau}$.

Let be a density matrix on $L^{2}(\mathbb{R}^{n})$. We assume that
$\widehat{\rho}=(2\pi\hbar)^{n}\operatorname*{Op}(\rho)$. The function $\rho$
is a real function on phase space $\mathbb{R}^{2n}$ and we have%
\begin{equation}
\operatorname*{Tr}\widehat{\rho}=\int_{\mathbb{R}^{2n}}\rho(z)dz=1.
\label{norm}%
\end{equation}
Observe that we do not in general have $\rho\geq0$.

Let $\widehat{A}$ be an observable. Its expectation value with respect to
$\rho$ is by definition the real number
\begin{equation}
\langle\widehat{A}\rangle=\int_{\mathbb{R}^{2n}}a(z)\rho(z)dz \label{avz}%
\end{equation}
where it is assumed that the integral on the right side is absolutely
convergent. We will write, with some abuse of notation,%
\begin{equation}
\langle\widehat{A}\rangle=\operatorname*{Tr}(\rho\widehat{A}) \label{abjtr}%
\end{equation}
(see the discussion in de Gosson \cite{Birkbis}, \S 12.3, of the validity of
various \textquotedblleft trace formulas\textquotedblright). In view of
formulas (\ref{aweyl}) and (\ref{opjaw}) we have either%
\[
\langle\widehat{A}\rangle=(2\pi\hbar)^{n}\sum_{j\geq1}\lambda_{j}\langle
a,\operatorname*{Wig}\psi_{j}\rangle
\]
(when $a\longleftrightarrow\widehat{A}$ is the Weyl correspondence) or%
\[
\langle\widehat{A}\rangle=(2\pi\hbar)^{n}\sum_{j\geq1}\lambda_{j}\langle
a,\operatorname*{Wig}\nolimits_{\mathrm{BJ}}\psi_{j}\rangle
\]
(when $a\longleftrightarrow\widehat{A}$ is the Born--Jordan correspondence)
and hence%
\begin{equation}
\langle\widehat{A}\rangle=(2\pi\hbar)^{n}\sum_{j\geq1}\lambda_{j}(\widehat
{A}\psi_{j}|\psi_{j})_{L^{2}}=\sum_{j\geq1}\lambda_{j}\langle\widehat
{A}\rangle_{j} \label{abjav}%
\end{equation}
where $\langle\widehat{A}\rangle_{j}=(\widehat{A}\psi_{j}|\psi_{j})_{L^{2}}$
(recall that $(\psi_{j})_{j\geq1}$ is an orthonormal system).

If $\widehat{A}^{2}$ also is an observable and if $\langle\widehat{A}%
^{2}\rangle=\operatorname*{Tr}(\rho\widehat{A}^{2})$ exists, then the number%
\begin{equation}
(\operatorname*{Var}\nolimits_{\rho}\widehat{A})^{2}=\langle\widehat{A}%
^{2}\rangle_{\rho}-\langle\widehat{A}\rangle_{\rho}^{2} \label{standard}%
\end{equation}
that is%
\begin{equation}
(\operatorname*{Var}\nolimits_{\rho}\widehat{A})^{2}=\operatorname*{Tr}%
(\widehat{\rho}\widehat{A}^{2})-\operatorname*{Tr}(\widehat{\rho}\widehat
{A})^{2} \label{trandard}%
\end{equation}
is the variance of $\widehat{A}$; its positive square root
$\operatorname*{Var}_{\rho}\widehat{A}$ is called \textquotedblleft standard
deviation\textquotedblright. More generally consider a second observable
$\widehat{B}$; then the covariance\ of the pair $(\widehat{A},\widehat{B})$
with respect to $\widehat{\rho}$ is defined by%
\begin{equation}
\operatorname*{Cov}\nolimits_{\rho}(\widehat{A},\widehat{B}%
)=\operatorname*{Tr}(\widehat{\rho}\widehat{A}\widehat{B})-\operatorname*{Tr}%
(\widehat{\rho}\widehat{A})\operatorname*{Tr}(\widehat{\rho}\widehat{B}).
\label{covdef}%
\end{equation}
It is in general a complex number, and we have%
\begin{equation}
\overline{\operatorname*{Cov}\nolimits_{\rho}(\widehat{A},\widehat{B}%
)}=\operatorname*{Cov}\nolimits_{\rho}(\widehat{B},\widehat{A}).
\label{covconj}%
\end{equation}
The covariance has the properties of a complex scalar product; it therefore
satisfies the Cauchy--Schwarz inequality%
\begin{equation}
\operatorname*{Tr}(\widehat{\rho}\widehat{A}^{2})\operatorname*{Tr}%
(\widehat{\rho}\widehat{B}^{2})\geq|\operatorname*{Cov}\nolimits_{\rho
}(\widehat{A},\widehat{B})|^{2}. \label{CS}%
\end{equation}
\ 

The following lemma will be useful in the proof of the uncertainty
inequalities below:

\begin{lemma}
If the covariance of two observables $\widehat{A}$ and $\widehat{B}$ exist we
have%
\begin{align}
\operatorname{Re}\operatorname*{Cov}\nolimits_{\rho}(\widehat{A},\widehat{B})
&  =\tfrac{1}{2}\operatorname*{Tr}(\widehat{\rho}\{\widehat{A},\widehat
{B}\})-\operatorname*{Tr}(\widehat{\rho}\widehat{A})\operatorname*{Tr}%
(\widehat{\rho}\widehat{B})\label{recov}\\
\operatorname{Im}\operatorname*{Cov}\nolimits_{\rho}(\widehat{A},\widehat{B})
&  =\tfrac{1}{2i}\operatorname*{Tr}(\widehat{\rho}[\widehat{A},\widehat{B}])
\label{imcov}%
\end{align}
where $\{\widehat{A},\widehat{B}\}=\widehat{A}\widehat{B}+\widehat{B}%
\widehat{A}$ and $[\widehat{A},\widehat{B}]=\widehat{A}\widehat{B}-\widehat
{B}\widehat{A}$ are, respectively, the anticommutator and the commutator of
$\widehat{A}$ and $\widehat{B}$.
\end{lemma}

\begin{proof}
We have, in view of (\ref{covconj}),
\begin{align*}
2\operatorname{Re}\operatorname*{Cov}\nolimits_{\rho}(\widehat{A},\widehat
{B})  &  =\operatorname*{Cov}\nolimits_{\rho}(\widehat{A},\widehat
{B})+\operatorname*{Cov}\nolimits_{\rho}(\widehat{B},\widehat{A})\\
&  =\operatorname*{Tr}(\widehat{\rho}\{\widehat{A},\widehat{B}%
\})-2\operatorname*{Tr}(\widehat{\rho}\widehat{A})\operatorname*{Tr}%
(\widehat{\rho}\widehat{B});
\end{align*}
for the second equality%
\begin{align*}
2i\operatorname{Im}\operatorname*{Cov}\nolimits_{\rho}(\widehat{A},\widehat
{B})  &  =\operatorname*{Cov}\nolimits_{\rho}(\widehat{A},\widehat
{B})-\operatorname*{Cov}\nolimits_{\rho}(\widehat{B},\widehat{A})\\
&  =\operatorname*{Tr}(\widehat{\rho}\widehat{A}\widehat{B}%
)-\operatorname*{Tr}(\widehat{\rho}\widehat{B}\widehat{A})\\
&  =\operatorname*{Tr}(\widehat{\rho}[\widehat{A},\widehat{B}]).
\end{align*}

\end{proof}

Observe that the anticommutator and commutator obey the relations
\[
\{\widehat{A},\widehat{B}\}^{\ast}=\{\widehat{A},\widehat{B}\}\text{ \ ,
\ }[\widehat{A},\widehat{B}]^{\ast}=-[\widehat{A},\widehat{B}]
\]
and hence $\langle\lbrack\widehat{A},\widehat{B}]\rangle$ is a pure imaginary
number or zero; in particular $|\langle\lbrack\widehat{A},\widehat{B}%
]\rangle|^{2}\leq0$.

\begin{proposition}
\label{propUP}If the variances and covariance of two observables $\widehat{A}$
and $\widehat{B}$ exist then:
\begin{equation}
(\operatorname*{Var}\nolimits_{\rho}\widehat{A})^{2}(\operatorname*{Var}%
\nolimits_{\rho}\widehat{B})^{2}\geq\operatorname*{Cov}\nolimits_{\rho
}^{\mathrm{sym}}(\widehat{A},\widehat{B})^{2}-\tfrac{1}{4}\langle
\lbrack\widehat{A},\widehat{B}]\rangle_{\rho}^{2} \label{principe1}%
\end{equation}
where $\langle\lbrack\widehat{A},\widehat{B}]\rangle^{2}<0$ and%
\begin{equation}
\operatorname*{Cov}\nolimits_{\rho}^{\mathrm{sym}}(\widehat{A},\widehat
{B})=\tfrac{1}{2}(\operatorname*{Cov}\nolimits_{\rho}(\widehat{A},\widehat
{B})+\operatorname*{Cov}\nolimits_{\rho}(\widehat{B},\widehat{A}))
\label{covsym}%
\end{equation}
is a real number. In particular the Heisenberg inequality%
\begin{equation}
(\operatorname*{Var}\nolimits_{\rho}\widehat{A})^{2}(\operatorname*{Var}%
\nolimits_{\rho}\widehat{B})^{2}\geq-\tfrac{1}{4}\langle\lbrack\widehat
{A},\widehat{B}]\rangle_{\rho}^{2} \label{principe1bis}%
\end{equation}
holds.
\end{proposition}

\begin{proof}
Replacing $\widehat{A}$ and $\widehat{B}$ with $\widehat{A}-\langle\widehat
{A}\rangle$ and $\widehat{B}-\langle\widehat{B}\rangle$ it is sufficient to
prove (\ref{principe1}) when $\langle\widehat{A}\rangle=\langle\widehat
{B}\rangle=0$. We thus have to prove the inequality%
\begin{equation}
\operatorname*{Tr}(\widehat{\rho}\widehat{A}^{2})\operatorname*{Tr}%
(\widehat{\rho}\widehat{B}^{2})\geq\operatorname*{Cov}\nolimits_{\rho
}^{\mathrm{sym}}(\widehat{A},\widehat{B})^{2}-\tfrac{1}{4}\operatorname*{Tr}%
(\widehat{\rho}[\widehat{A},\widehat{B}])^{2}. \label{principe2}%
\end{equation}
Noting that definition (\ref{covsym}) can be rewritten%
\[
\operatorname*{Cov}\nolimits_{\rho}^{\mathrm{sym}}(\widehat{A},\widehat
{B})=\tfrac{1}{2}\operatorname*{Tr}(\widehat{\rho}\{\widehat{A},\widehat
{B}\})-\operatorname*{Tr}(\widehat{\rho}\widehat{A})\operatorname*{Tr}%
(\widehat{\rho}\widehat{B}).
\]
we remark that in view of formulas (\ref{recov}) and (\ref{imcov}) in the
lemma above we have%
\begin{align}
|\operatorname*{Cov}\nolimits_{\rho}(\widehat{A},\widehat{B})|^{2}  &
=\tfrac{1}{4}\operatorname*{Tr}(\widehat{\rho}\{\widehat{A},\widehat{B}%
\})^{2}-\tfrac{1}{4}\operatorname*{Tr}(\widehat{\rho}[\widehat{A},\widehat
{B}])^{2}\label{covtr}\\
&  =\operatorname*{Cov}\nolimits_{\rho}^{\mathrm{sym}}(\widehat{A},\widehat
{B})^{2}-\tfrac{1}{4}\operatorname*{Tr}(\widehat{\rho}[\widehat{A},\widehat
{B}])^{2}%
\end{align}
hence the proof of (\ref{principe2}) is reduced to the proof of the inequality%
\begin{equation}
\operatorname*{Tr}(\widehat{\rho}\widehat{A}^{2})\operatorname*{Tr}%
(\widehat{\rho}\widehat{B}^{2})\geq|\operatorname*{Cov}\nolimits_{\rho
}(\widehat{A},\widehat{B})|^{2} \label{principe3}%
\end{equation}
which is just the Cauchy--Schwarz inequality (\ref{CS}) for covariances.
\end{proof}

\subsection{Weyl vs Born--Jordan}

Let us discuss the similarities and differences between the uncertainty
principles associated with the Weyl and Born--Jordan correspondences. First,
as already observed in the Introduction, the Weyl and Born--Jordan
quantizations of monomials $x_{j}^{m}p_{j}^{n}$ (and hence of their linear
combinations) are identical when $m+n\leq2$. This implies, in particular, that
if the symbols $a$ and $b$ are, respectively, multiplication by the
coordinates $x_{j}$ and $p_{j}$ then the corresponding operators$\ \widehat
{A}$ and $\widehat{B}$ are, in both cases given by $\widehat{X}_{j}=x_{j}$ and
$\widehat{P}_{j}=-i\hbar\partial_{x_{j}}$. It follows that
$\operatorname*{Var}_{\rho}\widehat{X}_{j}$ and $\operatorname*{Var}_{\rho
}\widehat{P}_{j}$ satisfy the usual Robertson--Schr\"{o}dinger inequalities%
\begin{equation}
\operatorname*{Var}\nolimits_{\rho}\widehat{X}_{j}\operatorname*{Var}%
\nolimits_{\rho}\widehat{P}_{j}\geq\operatorname*{Cov}\nolimits_{\rho
}^{\mathrm{sym}}(\widehat{X}_{j},\widehat{P}_{j})^{2}+\tfrac{1}{4}\hbar^{2}.
\label{RS1}%
\end{equation}

We mention that the inequalities (\ref{RS1}) can be rewritten in compact form
as%
\begin{equation}
\Sigma+\tfrac{1}{2}i\hbar J\geq0\label{sigmaj}%
\end{equation}
where $\geq0$ means \textquotedblleft semi-definite positive\textquotedblright%
, $J$ is the standard symplectic matrix, and
\begin{equation}
\Sigma=%
\begin{pmatrix}
\operatorname*{Cov}\nolimits_{\rho}(\widehat{X},\widehat{X}) &
\operatorname*{Cov}\nolimits_{\rho}(\widehat{X},\widehat{P})\\
\operatorname*{Cov}\nolimits_{\rho}(\widehat{P},\widehat{X}) &
\operatorname*{Cov}\nolimits_{\rho}(\widehat{P},\widehat{P})
\end{pmatrix}
\label{covmat}%
\end{equation}
is the statistical covariance matrix. The formulation (\ref{sigmaj}) of the
Robertson--Schr\"{o}dinger inequalities clearly shows one of the main
features, namely the symplectic covariance of these inequalities, which we
have used in previous work \cite{FP,physreps} to express the uncertainty
principle in terms of the notion of symplectic capacity, which is closely
related to Gromov's non-squeezing theorem from symplectic topology. This has
also given us the opportunity to discuss the relations between classical and
quantum mechanics in \cite{gohi}.

Let $\widehat{S}\in\operatorname*{Mp}(2n,\mathbb{R})$, $S=\pi
^{\operatorname*{Mp}}(\widehat{S})$ and set $\widehat{A}^{\prime}=\widehat
{S}\widehat{A}\widehat{S}^{-1}$, $\widehat{B}^{\prime}=\widehat{S}\widehat
{B}\widehat{S}^{-1}$; we are assuming that $\widehat{A},\widehat{B}$
correspond, as in the proof of Proposition \ref{propUP}, to an arbitrary
quantization scheme $a\longleftrightarrow\widehat{A}$. We have quite
generally, using the cyclicity of the trace,%
\begin{align*}
\operatorname*{Cov}\nolimits_{\rho}^{\mathrm{sym}}(\widehat{A},\widehat{B})
&  =\tfrac{1}{2}\operatorname*{Tr}(\widehat{\rho}\{\widehat{A},\widehat
{B}\})-\operatorname*{Tr}(\widehat{\rho}\widehat{A})\operatorname*{Tr}%
(\widehat{\rho}\widehat{B})\\
&  =\tfrac{1}{2}\operatorname*{Tr}(\widehat{\rho}\widehat{S}^{-1}\{\widehat
{A}^{\prime},\widehat{B}^{\prime}\}\widehat{S})-\operatorname*{Tr}%
(\widehat{\rho}\widehat{S}^{-1}\widehat{A}^{\prime}\widehat{S}%
)\operatorname*{Tr}(\widehat{\rho}\widehat{S}^{-1}\widehat{B}^{\prime}%
\widehat{S})\\
&  =\tfrac{1}{2}\operatorname*{Tr}(\widehat{S}^{-1}\widehat{\rho}\widehat
{S}\{\widehat{A}^{\prime},\widehat{B}^{\prime}\})-\operatorname*{Tr}%
(\widehat{S}^{-1}\widehat{\rho}\widehat{S}\widehat{A}^{\prime}%
)\operatorname*{Tr}(\widehat{S}\widehat{\rho}\widehat{S}^{-1}\widehat
{B}^{\prime})\\
&  =\operatorname*{Cov}\nolimits_{\widehat{S}\widehat{\rho}\widehat{S}^{-1}%
}^{\mathrm{sym}}(\widehat{A}^{\prime},\widehat{B}^{\prime});
\end{align*}
similarly $\operatorname*{Var}\nolimits_{\rho}\widehat{A}=\operatorname*{Var}%
\nolimits_{\widehat{S}\widehat{\rho}\widehat{S}^{-1}}\widehat{A}$ and
$\langle\lbrack\widehat{A},\widehat{B}]\rangle_{\rho}^{2}=\langle
\lbrack\widehat{A}^{\prime},\widehat{B}^{\prime}]\rangle_{\widehat{S}%
\widehat{\rho}\widehat{S}^{-1}}^{2}$ and hence%
\[
(\operatorname*{Var}\nolimits_{\rho^{\prime}}\widehat{A}^{\prime}%
)^{2}(\operatorname*{Var}\nolimits_{\rho^{\prime}}\widehat{B}^{\prime}%
)^{2}\geq\operatorname*{Cov}\nolimits_{\rho^{\prime}}^{\mathrm{sym}}%
(\widehat{A}^{\prime},\widehat{B}^{\prime})^{2}-\tfrac{1}{4}\langle
\lbrack\widehat{A}^{\prime},\widehat{B}^{\prime}\rangle_{\rho^{\prime}}^{2}.
\]
with $\rho^{\prime}=\widehat{S}\widehat{\rho}\widehat{S}^{-1}$. Suppose now
that the operator correspondence $a\longleftrightarrow\widehat{A}$ is the Weyl
correspondence; then, by Proposition \ref{propams} we have $\widehat{S}%
\rho\widehat{S}^{-1}=Op(\rho\circ s^{-1})$ and the inequalities
(\ref{principe1}) become%
\begin{equation}
(\operatorname*{Var}\nolimits_{\rho\circ s^{-1}}A_{\mathrm{W}}^{\prime}%
)^{2}(\operatorname*{Var}\nolimits_{\rho\circ s^{-1}}B_{\mathrm{W}}^{\prime
})^{2}\geq\operatorname*{Cov}\nolimits_{\rho\circ s^{-1}}^{\mathrm{sym}%
}(A_{\mathrm{W}}^{\prime},B_{\mathrm{W}}^{\prime})^{2}-\tfrac{1}{4}%
\langle\lbrack A_{\mathrm{W}}^{\prime},B_{\mathrm{W}}^{\prime}]\rangle
_{\rho\circ s^{-1}}^{2}. \label{covupW}%
\end{equation}
Again, in view of Proposition \ref{propams}, in the Born--Jordan case we have
inequality%
\begin{equation}
(\operatorname*{Var}\nolimits_{\rho\circ s^{-1}}A_{\mathrm{BJ}}^{\prime}%
)^{2}(\operatorname*{Var}\nolimits_{\rho\circ s^{-1}}B_{\mathrm{BJ}}^{\prime
})^{2}\geq\operatorname*{Cov}\nolimits_{\rho\circ s^{-1}}^{\mathrm{sym}%
}(A_{\mathrm{BJ}}^{\prime},B_{\mathrm{BJ}}^{\prime})^{2}-\tfrac{1}{4}%
\langle\lbrack A_{\mathrm{BJ}}^{\prime},B_{\mathrm{BJ}}^{\prime}]\rangle
_{\rho\circ s^{-1}}^{2} \label{covupBJ}%
\end{equation}
only for those $\widehat{S}\in\operatorname*{Mp}(2n,\mathbb{R})$ which are
products of metaplectic operators of the type $\widehat{J}$ and $\widehat
{M}_{L,m}$.

\section{Discussion}

There is an old ongoing debate in quantum mechanics on which quantization
scheme is the most adequate for physical applications; an interesting recent
contribution is that of Kauffmann \cite{kauffmann}, who seems to favor the
Born--Jordan correspondence. The introduction of the $\tau$-Wigner and
Born--Jordan distributions has been motivated in time-frequency analysis by
the fact that the usual cross-Wigner distribution gives raise to disturbing
ghost frequencies; it was discovered by Boggiatto and his collaborators
\cite{bogetal,bogetalbis,bogetalter} that these ghost frequencies were
attenuated by averaging over $\tau$.

The study of uncertainties for non-standard situations has been tackled (from
a very different point of view) by Korn \cite{korn}; also see the review paper
\cite{fosi} by Folland and Sitaram, which however unfortunately deliberately
ignores the fundamental issue of covariance. Gibilisco and his collaborators
\cite{gibilisco,gibb0,gibbis} give highly nontrivial refinements of
uncertainty relations using convexity properties, and studied the notion of
statistical covariance in depth.

We mention that in a very well written thesis, published as a book, Steiger
\cite{steiger} has given an interesting historical review and analysis of the
evolution of the uncertainty principle; in addition he compares the interest
\ of several different formulations, and gives a clever elementary derivation
of the Robertson--Schr\"{o}dinger inequalities for operators. The work also
contains Matematica codes for the computation of (co-)variances.

\begin{acknowledgement}
This work has been supported by a research grant from the Austrian Research
Agency FWF (Projektnummer P23902-N13).
\end{acknowledgement}


\begin{thebibliography}{99}                                                                                               %


\bibitem {bogetal}P. Boggiatto, G. De Donno, A. Oliaro, Time-Frequency
Representations of Wigner Type and Pseudo-Differential Operators, T. Am. Math.
Soc. 362(9) (2010) 4955--4981.

\bibitem {bogetalbis}P. Boggiatto, Bui Kien Cuong, G. De Donno, A. Oliaro,
Weighted integrals of Wigner representations, J.Pseudo-Differ. Oper. Appl.
1(4) (2010) 401--415.

\bibitem {bogetalter}P. Boggiatto, G. De Donno, A. Oliaro. Hudson Theorem for
$\tau$-Wigner Transforms, Quaderni scientifici del Dipartimento di Matematica,
Universit\`{a} di Torino, Quaderno N. 1 (2010).

\bibitem {bj}M. Born, P. Jordan, Zur Quantenmechanik, Z. Physik 34, 858--888 (1925).

\bibitem {bjh}M. Born, W. Heisenberg, P. Jordan, Zur Quantenmechanik II, Z.
Physik 35 (1926) 557--615.

\bibitem {ca78}L. Castellani, Quantization Rules and Dirac's Correspondence,
Il Nuovo Cimento 48A(3) (1978) 359--368.

\bibitem {cr89}P. Crehan, The parametrisation of quantisation rules equivalent
to operator orderings, and the effect of different rules on the physical
spectrum, J. Phys. A:\ Math. Gen. 22 (1989) 811--822.

\bibitem {dirac1}P.A.M. Dirac, The fundamental equations of quantum mechanics,
Proc. R. Soc. London, Ser. A, 109 (1925) 642--653.

\bibitem {fepr09}W.A. Fedak, J.J. Prentis, The 1925 Born and Jordan paper
\textquotedblleft On quantum mechanics\textquotedblright, Am. J. Phys. 77(2)
(2009) 128--139.

\bibitem {Folland}G.B. Folland, Harmonic Analysis in Phase space, Annals of
Mathematics studies, Princeton University Press, Princeton, N.J., 1989.

\bibitem {fosi}G.B. Folland, A. Sitaram, The uncertainty principle: a
mathematical survey, J. Fourier Anal. Appl. 3(3) (2007) 207--238

\bibitem {gibilisco}P. Gibilisco, D. Imparato, and T. Isola, Uncertainty
principle and quantum Fisher information II, J. Math. Phys. 48, 072109 (2007).

\bibitem {gibb0}P. Gibilisco, F. Hiai, D. Petz, Quantum covariance, quantum
Fisher information, and the uncertainty relations, IEEE Trans. Inform. Theory
55(1) (2009) 439--443.

\bibitem {gibbis} P. Gibilisco, T. Isola, How to distinguish quantum
covariances using uncertainty relations. J. Math. Anal. Appl. 384(2) (2011)
2, 670--676

\bibitem {Birk}M. de Gosson, Symplectic Geometry and Quantum Mechanics,
Birkh\"{a}user, Basel, 2006.

\bibitem {Birkbis}M. de Gosson, Symplectic Methods in Harmonic Analysis;
Applications to Mathematical Physics, Birkh\"{a}user, 2011.

\bibitem {FP}M. de Gosson, The symplectic camel and the uncertainty principle:
the tip of an iceberg? Found. Phys. 39(2) (2009) 194--214

\bibitem {transam}M. de Gosson, Symplectic Covariance Properties for Shubin
and Born--Jordan Pseudo-Differential Operators. T. Am. Math. Soc. (2012).

\bibitem {gohi}M. de Gosson, B. Hiley, Imprints of the Quantum World in
Classical Mechanics, Found. Phys. 41(9) (2011) 1415--1436.

\bibitem {physreps}M. de Gosson, F. Luef, Symplectic capacities and the
geometry of uncertainty: the irruption of symplectic topology in classical and
quantum mechanics. Phys. Rep. 484(5) (2009) 131--179. 

\bibitem {golu1}M. de Gosson, F. Luef, Preferred Quantization Rules:
Born--Jordan vs. Weyl; Applications to Phase Space Quantization.
J.Pseudo-Differ. Oper. Appl. 2(1) (2011) 115--139.

\bibitem {Gro}K. Gr\"{o}chenig, Foundations of Time-Frequency Analysis,
Birkh\"{a}user, Boston, 2000.

\bibitem {heisenberg}W. Heisenberg, \"{U}ber quantentheoretishe Umdeutung
kinematisher und mechanischer Beziehungen, Zeitschrift f\"{u}r Physik, 33
(1925) 879--893; English translation in: B.L. van der Waerden, editor, Sources
of Quantum Mechanics, Dover Publications,1968.

\bibitem {Hudson}R.L. Hudson, When is the Wigner quasi-probability density
non-negative? Rep. Math. Phys. 6 (1974), 249--252.

\bibitem {kauffmann}S.K. Kauffmann, Unambiguous quantization from the maximum
classical correspondence that is self-consistent: the slightly stronger
canonical commutation rule Dirac missed, Found. Phys. 41 (2011) 805--918.

\bibitem {korn}P. Korn, Some Uncertainty Principles for Time-Frequency
Transforms of the Cohen Class, IEEE T. Signal. Proces. 53(2) (2005) 523--527

\bibitem {mccoy}N.H. McCoy, On the function in quantum mechanics which
corresponds to a given function in classical mechanics,Proc. Natl. Acad. Sci.
U.S.A. 18(11) (1932) 674--676.

\bibitem {scholz}E. Scholz, Weyl entering the `new' quantum mechanics
discourse, C. Joas, C. Lehner, J. Renn (eds.). HQ-1: Conference on the History
of Quantum Physics (Berlin, July 2--6, 2007), Preprint MPI History of Science
Berlin, 350 vol. II (2007).

\bibitem {sch}E. Schr\"{o}dinger, Zum Heisenbergschen Unsch\"{a}rfeprinzip,
Sitzungsberichte der Preu\ss ischen Akademie der Wissenschaften.
Physikalisch-mathematische Klasse (1930) 296--303.

\bibitem {sh87}M. A. Shubin, Pseudodifferential Operators and Spectral Theory,
Springer-Verlag, 1987; original Russian edition in Nauka, Moskva, 1978.

\bibitem {steiger}N.J. Steiger, Quantum Uncertainty and Conservation Law
Restrictions on Gate Fidelity, Brigham Young University. Department of Physics
and Astronomy, 2010.

\bibitem {jvn}J. von Neumann, Wahrscheinligkeitstheoretischer Aufbau der
Quantenmechanik, G\"{o}ttinger Nachrichten 1 (1927) 245--272.

\bibitem {Weyl}H. Weyl, Quantenmechanik und Gruppentheorie, Zeitschrift
f\"{u}r Physik, 46 (1927).

\bibitem {WeylRob}H. Weyl, The Theory of Groups and Quantum Mechanics,
translated from the 2nd German edition by H.P. Robertson, New York, Dutten (1931).
\end{thebibliography}
\end{document}